\documentclass[journal,10pt]{IEEEtran}
\usepackage[T1]{fontenc}
\usepackage{amsmath,amssymb,amsthm}
\usepackage{graphicx,cite}
\usepackage[hidelinks]{hyperref}
\usepackage{microtype}
\newtheorem{theorem}{Theorem}
\newtheorem{lemma}[theorem]{Lemma}
\newtheorem{proposition}[theorem]{Proposition}
\newtheorem{corollary}[theorem]{Corollary}
\newcommand{\E}{\mathbb{E}}
\newcommand{\Prob}{\mathbb{P}}
\newcommand{\R}{\mathbb{R}}
\newcommand{\Z}{\mathbb{Z}}
\newcommand{\N}{\mathbb{N}}
\newcommand{\ind}{\mathbf{1}}
\newcommand{\dist}{\operatorname{dist}}
\newcommand{\diam}{\operatorname{diam}}

\title{Limit Theorems and Optimal Policies for Flock\\ Flight through Poisson Forests}
\author{Lav R. Varshney\thanks{L. R. Varshney is with the AI Innovation Institute and the Department of Electrical and Computer Engineering, Stony Brook University, Stony Brook, NY, USA.}}
\begin{document}
\maketitle
\begin{abstract}
Inspired by birds flying through cluttered environments, Karaman and Frazzoli showed that high-speed flight through an ergodic random forest exhibits an ultimate speed limit: above a critical speed, almost surely no infinite collision-free trajectory exists for a single bird. This paper formulates a corresponding bulk-transport problem for many birds. We construct an explicit directed corridor network in a Poisson disk forest and impose a common temporal headway at its waypoints. The maximum current equals a minimum vertex-cut capacity and is achieved by repeated flight along vertex-disjoint paths. The same schedule maximizes completed traversals at every finite horizon, with constant forward speed and no internal waiting. For a forest belt of fixed depth, exact transverse superadditivity yields an almost-sure deterministic current-density limit. A uniform shape theorem for an asymmetric planar dual identifies the current when depth and width grow together. For diverging depth bounded by a constant times width, this current equals the width-first, depth-second limit and is independent of the positive aspect ratio. When diagonal corridors are available, bulk current is positive at sufficiently low obstacle intensity and zero at sufficiently high intensity. For stochastic demand, an upstream admission queue is stable at every rate below the finite-network capacity.
\end{abstract}

\section{Introduction}
The problem of high-speed flight through clutter is an example of how geometry, dynamics, and uncertainty impose fundamental performance limits. Karaman and Frazzoli considered a bird or bird-like vehicle flying through a randomly generated planar forest and proved a phase transition result: for an ergodic forest, there is a critical speed above which every online planning policy eventually has tree collisions almost surely, whereas below the critical speed an infinite-length collision-free trajectory exists almost surely under suitable conditions~\cite{karaman2012}. In the homogeneous Poisson disk case, their construction connects the motion-planning problem to percolation theory.

That result is a single-particle theorem. It asks whether at least one infinite-length feasible trajectory exists. A different question arises for flock motion, swarms, or high-density traffic through clutter: how many agents can traverse the forest per unit time? This is not determined by the single-bird speed limit directly. One feasible path can be reused, but a finite body size and a separation requirement limit the rate of reuse. At the same time, a random medium may contain many parallel corridors and enable extensive current flow. The bulk problem is therefore a transport problem through a random medium.

The formulation directly draws on Karaman and Frazzoli in that the forest geometry and vehicle dynamics induce a directed feasibility structure. Methodologically, however, we build on transport through bulk materials, considering a steady current constrained by void geometry, analogous in motivation to transport through porous media~\cite{bear1972,chayes1986}. Maximal-flow theory identifies bottlenecks through max-flow/min-cut duality~\cite{ford1956,elias1956}. Laws of large numbers for maximal flows in random media provide related asymptotic results~\cite{kesten1987,rossignol2010tilted,rossignol2010box}. Note that geometric max-flow methods have also been used to determine navigation lanes with clearance~\cite{farias2020}.

More specifically, the formulation here considers a homogeneous Poisson disk forest and an explicit library of noncrossing forward segments together with a common waypoint headway. This gives a direct connection between a discrete cut and a physical schedule. Uniform waypoint capacities imply that an integral maximum flow decomposes into vertex-disjoint paths. The embedding separates different paths in space whereas a common headway separates consecutive birds on the same path in time. Thus the finite-network capacity is achieved without stopping inside the forest. Since every route has the same transit time, this schedule also maximizes cumulative deliveries at every horizon. The repeated-path construction is related to classical dynamic network flow~\cite{ford1958}, with a separate geometric argument providing the collision-avoidance property. The resulting current is optimal within the prescribed corridor-and-headway protocol and provides an achievable lower bound for unrestricted continuous flight.

We first consider a forest belt of fixed longitudinal depth $L$ and increasing transverse width $W$. This models passage through a finite obstacle field rather than indefinite navigation. A finite belt may support positive current even at a speed for which no infinite trajectory exists. For that setting, we identify a deterministic current per unit width, prove that it is achievable by collision-free forward flight, and give an admission policy for stochastic demand.

In the asymptotic regime, optimal path families in adjacent transverse bands can be combined, giving exact superadditivity. Kingman's ergodic theorem~\cite{kingman1973,steele1989} then gives a deterministic limit. A reverse inequality bounds the benefit of crossing a band interface by the number of waypoints along that interface. The Poisson void probability gives a positive lower bound, and bounded differences control fluctuations around the mean. Under the stated arrival-moment and dispatch assumptions, an upstream admission policy is stable at every rate below the finite-network capacity~\cite{neely2010}.

A second scaling question arises when depth and width grow together, getting to the Karaman--Frazzoli requirement of infinite-length paths without crashing. We address this setting without changing the flight protocol. Zero probability of an infinite open corridor path forces zero normalized current whenever depth and width both diverge. A bounded planar representation of the waypoint capacities identifies cuts with top-to-bottom paths in an asymmetric dual. Deterministic routing bounds and ergodicity yield a shape theorem uniform over moving endpoints. Reflection symmetry makes the vertical direction minimize the limiting cut cost. A finite-block confinement argument then identifies this cost with the width-first, depth-second current. Depth monotonicity extends the identification to every diverging depth sequence of order at most the width. In particular, the bulk current is independent of the positive aspect ratio.

The common limit is strictly positive for sufficiently small obstacle intensity when diagonal corridors are available. The proof uses a diagonal subnetwork and counts dual cuts with many blocked corridors. A finite coloring of their clearance regions accounts for dependence. At sufficiently large intensity, a forward-path count proves extinction. Thus extensive current has a nontrivial intensity threshold within the prescribed protocol. The order of spatial limits remains important: at each fixed finite width, a sufficiently deep forest is eventually blocked almost surely.

The remainder of the paper is organized as follows. Sections~\ref{sec:model}--\ref{sec:finite} specify the protocol and prove its finite-network capacity. Sections~\ref{sec:transverse}--\ref{sec:depth} establish the transverse law and its depth and speed dependence. Admission control follows in Section~\ref{sec:queue}. Section~\ref{sec:bulk} identifies the joint-growth limit, and Section~\ref{sec:intensity} gives positive-current and extinction regimes.

\section{Forest, Dynamics, and Corridor Protocol}\label{sec:model}
Let us introduce the basic mathematical model.
\subsection{Poisson Forest and Forward Flight}
Let $\{Y_k\}$ be a homogeneous Poisson point process in $\R^2$ with intensity $\zeta\in(0,\infty)$, and let every tree have radius $r>0$. The obstacle set is
\begin{equation}
\mathcal O=\bigcup_k B(Y_k,r).\label{eq:obstacles}
\end{equation}
The standard Poisson properties used in the sequel are stationarity, independence on disjoint regions, and the void probability $\Prob\{\#(\{Y_k\}\cap A)=0\}=e^{-\zeta|A|}$ for bounded Borel sets $A$~\cite{last2017}. In particular, translation by any fixed nonzero vector is mixing and ergodic: events determined by bounded regions become independent after sufficiently distant translations, and approximation by such events extends mixing to the full process.

A bird is a disk of radius $a>0$. Fix a clearance margin $\eta>0$ and let $q=a+\eta$. The center follows an absolutely continuous trajectory in two-dimensional state space $(x_1,x_2)$, with
\begin{equation}
\dot x_1(t)=\nu,\quad |\dot x_2(t)|\le u,\quad \nu,u>0,\label{eq:dynamics}
\end{equation}
almost everywhere throughout passage inside the forest. This convention permits the piecewise-linear trajectories used below, whose lateral velocity changes at waypoints.

For a constant $q$, choose pitches $d,b>0$ and a temporal headway $H>0$ satisfying
\begin{equation}
\min\{d,b\}>2q,\quad \nu H>2q.\label{eq:spacing}
\end{equation}
The first inequality provides spatial separation between different lanes whereas the second provides longitudinal separation between successive birds on one lane.

\subsection{Explicit Forward Corridor Library}
The deterministic waypoint set is
\begin{equation}
z_{i,j}=\bigl(id,b(j+\tfrac12)\bigr),\quad i,j\in\Z,\quad i+j\text{ even}.\label{eq:waypoints}
\end{equation}
At even $i$, include the straight arc
\begin{equation}
(i,j)\longrightarrow(i+2,j).\label{eq:straight}
\end{equation}
At every $i$, also include the diagonal arcs
\begin{equation}
(i,j)\longrightarrow(i+1,j\pm1)\label{eq:diagonal}
\end{equation}
when
\begin{equation}
\nu b/d\le u.\label{eq:availability}
\end{equation}
Indices denote their associated waypoints in~\eqref{eq:waypoints}. Each arc is embedded as its straight line segment. Straight arcs advance by $2d$ with no lateral speed. Diagonal arcs advance by $d$ and require lateral speed $\nu b/d$. The graph is acyclic and embedded without crossings except at common endpoints. Every directed path from $x_1=0$ to $x_1=2Nd$ visits each intervening even longitudinal layer. The restriction of straight arcs to even layers ensures this last property.

For an embedded arc $e$, let $[e]$ denote its closed centerline segment. It is open when
\begin{equation}
\dist(Y_k,[e])>R\text{ for every }k,\quad R=r+q.\label{eq:open}
\end{equation}
Thus the entire swept tube of radius $q$ avoids the forest. Openness is defined in the full plane before any band is selected. Restricting the network to a smaller band does not change an arc's status. Note that neighboring openness variables need not be independent, since their clearance regions overlap. Fig.~\ref{fig:forest} illustrates the resulting corridor structure.

\begin{figure}[t]
\centering
\includegraphics[width=\columnwidth]{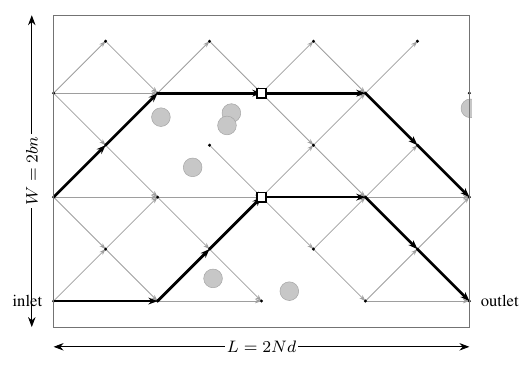}
\caption{An illustrative Poisson forest and its open corridor graph. Gray disks are trees; light solid segments are open corridors. Heavy paths form a maximum vertex-disjoint family, and squares mark a minimum waypoint cut. Blocked corridors are omitted. The computation gives $K=2$, so repeated launches attain current $2/H$. The square symbols identify a capacity obstruction.}
\label{fig:forest}
\end{figure}

Fix $N\in\N$ and set $L=2Nd$. For integers $m<n$, retain the open arcs and waypoints with
\begin{equation}
0\le i\le 2N,\quad 2m\le j<2n.\label{eq:band}
\end{equation}
Call this finite graph $G_N[m,n]$. Its centerlines lie in the transverse band $[2bm,2bn]$, whose width is $2b(n-m)$. The half-row offset in~\eqref{eq:waypoints} leaves at least $b/2>q$ of lateral clearance from the band boundary. No corridor leaves through the upper or lower side. The inlet and outlet sets are the waypoints at $i=0$ and $i=2N$: there are $n-m$ of each.

\subsection{Admissible Traffic}
To develop a finite bulk-current limit, we need a capacity mechanism in addition to single-bird feasibility. If birds were point particles and did not have mutual exclusion, they could follow each other such that one admissible trajectory is reused at arbitrarily small time offsets. Finite body size, clearance, and the following waypoint rule give the current a definite operational meaning. Birds are admitted at inlet waypoints, traverse open graph paths, and leave at the outlet. They may not wait inside the forest, leave the corridor library, or collide with one another. At every waypoint, including the inlet and outlet, successive passage times of different birds must differ by at least $H$. Together with the constant forward speed and prescribed corridor spacing, this common waypoint rule enforces body exclusion. Full map information and dispatchable access to the inlet waypoints are available. Any waiting for admission takes place outside the forest.

Let $D_\pi(T)$ be the number of birds delivered by time $T$ under an admissible execution $\pi$, starting with an empty forest. Define the maximum sustainable current by
\begin{equation}
J_N[m,n]=\sup_\pi\liminf_{T\to\infty}\frac{D_\pi(T)}{T}.\label{eq:current}
\end{equation}
We will bound cumulative deliveries at every horizon and give a single schedule attaining all these bounds.

Note that this formulation has a prescribed corridor-and-headway structure. Its optimum is not asserted to be extremal over all continuous collision-free trajectories. Every schedule given in the sequel is, however, a valid continuous flight under~\eqref{eq:dynamics}, so its current is also an achievable lower bound for the unrestricted motion problem with the same obstacles and bird size.

\section{Exact Finite-Network Capacity and Flight Schedule}\label{sec:finite}
\subsection{Vertex Cuts and Integral Flows}
Let $K_N[m,n]$ be the maximum number of pairwise vertex-disjoint inlet-to-outlet paths in $G_N[m,n]$. Endpoints are included in the disjointness requirement. To compute $K_N[m,n]$, split every waypoint $v$ into $v^-$ and $v^+$ with a capacity-one arc $v^-\to v^+$. Replace every open corridor $v\to w$ by $v^+\to w^-$, and add a supersource and supersink. Give all corridor and terminal-attachment arcs capacity $n-m+1$.

The inlet vertex arcs give a cut of capacity $n-m$, so no minimum cut uses an arc of capacity $n-m+1$. Max-flow integrality and path decomposition therefore give
\begin{equation}
K_N[m,n]=\min_C|C|,\quad 0\le K_N[m,n]\le n-m,\label{eq:cut}
\end{equation}
where $C$ ranges over waypoint sets meeting every inlet-to-outlet path. This is the vertex-capacitated form of max-flow/min-cut~\cite{ford1956,ahuja1993}. Dividing all split-network capacities by $H$ gives current capacity $K_N[m,n]/H$. The primal problem asks how much bird current can be routed; the dual problem asks for the weakest separating bottleneck. Notice that a minimum cut need not lie on a straight transverse cross-section.

\subsection{Geometric Separation}
Now we invoke the geometry of the problem.
\begin{lemma}[separation of disjoint paths]\label{lem:separation}
Two embedded candidate arcs with no common endpoint have Euclidean distance at least $\min\{d,b\}$. Consequently, the centerlines of two vertex-disjoint paths have distance greater than $2q$ under~\eqref{eq:spacing}.
\end{lemma}
\begin{IEEEproof}
First take $d=b=1$ and ignore the common vertical translation by $1/2$. Diagonal segments have slopes $\pm1$, integer endpoints of even coordinate sum, and length $\sqrt2$. Distinct parallel supporting lines differ by an even intercept, and perpendicular supporting lines intersect at a lattice waypoint. Nonincident diagonal segments therefore have distance at least $\sqrt2$. Horizontal segments lie at even heights, have even longitudinal endpoints, and have length two. Nonincident horizontal segments have distance at least two. The distance between a horizontal and a nonincident diagonal segment is at least one: an odd-height diagonal endpoint is at least one unit from the horizontal line, while an even-height endpoint on that line could meet the horizontal segment only at a common waypoint. The nearest remaining endpoint or segment separation is no smaller. Thus the minimum over all nonincident pairs is one. The linear map $(x,y)\mapsto(dx,by)$ multiplies every distance by at least $\min\{d,b\}$. Segments belonging to vertex-disjoint paths have no common endpoint, and so~\eqref{eq:spacing} finishes the proof.
\end{IEEEproof}

\subsection{A Matching Converse and Achievability Theorem}
Every crossing has the same transit time $\tau=L/\nu$. This permits an exact delivery bound as well as a current bound.
\begin{theorem}[exact deliveries without internal waiting]\label{thm:finite}
For every forest realization and every finite band, let $K=K_N[m,n]$ and $\tau=L/\nu$. Then
\begin{equation}
\sup_\pi D_\pi(T)=
\begin{cases}
0,&0\le T<\tau,\\
K\bigl(\lfloor(T-\tau)/H\rfloor+1\bigr),&T\ge\tau.
\end{cases}\label{eq:finitehorizon}
\end{equation}
A maximum vertex-disjoint path family, with one bird launched on each path at times $0,H,2H,\ldots$, attains these bounds simultaneously. In particular,
\begin{equation}
J_N[m,n]=\frac{K_N[m,n]}{H}.\label{eq:exactcurrent}
\end{equation}
\end{theorem}
\begin{IEEEproof}
Let $C$ be a minimum waypoint cut, so $|C|=K$. No bird can be delivered before $\tau$. For $T\ge\tau$, a bird delivered by $T$ can visit a cut vertex $v$ only during
\begin{equation}
\left[\frac{x_1(v)}\nu,\ T-\frac{L-x_1(v)}\nu\right].\label{eq:cutwindow}
\end{equation}
The left endpoint follows from the initially empty forest, and the right endpoint from the remaining transit time. This interval has length $T-\tau$. The waypoint-headway rule permits at most $\lfloor(T-\tau)/H\rfloor+1$ such visits. Assign each delivered bird to its first visited cut vertex and sum over $C$. This proves the upper bound in~\eqref{eq:finitehorizon}. When $K=0$, no delivery is possible.

For $K>0$, choose $K$ vertex-disjoint open paths. On each path, launch birds at times $0,H,2H,\ldots$ and parameterize motion by longitudinal position. Every segment obeys~\eqref{eq:dynamics}, and every path has transit time $\tau$. Openness gives tree clearance. At a common global time, two birds on the same path that were launched $h$ time units apart have longitudinal separation $\nu h\ge\nu H>2q$, whenever both are in the forest. Birds on different paths remain more than $2q$ apart by Lemma~\ref{lem:separation}. Each waypoint belongs to only one selected path, and its successive passage times differ by exactly $H$. Thus all collision constraints hold without internal waiting. The schedule yields the right side of~\eqref{eq:finitehorizon}; dividing by $T$ proves~\eqref{eq:exactcurrent}.
\end{IEEEproof}

This theorem does not assume independent-edge scheduling. Competing routes may share waypoints in the available network, but an integral optimum selects paths that do not. The geometric embedding then removes unmodeled conflicts between the selected paths. Exact attainment is possible here because every waypoint has the same headway capacity and every crossing has the same transit time. Arbitrary real capacities or arbitrary crossing corridor libraries would require a different scheduling argument.

\section{Transverse Current-Density Limit}\label{sec:transverse}
Consider the current density limit.
\subsection{A Pathwise Band Inequality}
\begin{lemma}[superadditivity and interface cost]\label{lem:bands}
For integers $m<k<n$, every forest realization satisfies
\begin{align}
K_N[m,k]+K_N[k,n]&\le K_N[m,n],\label{eq:super}\\
K_N[m,n]&\le K_N[m,k]+K_N[k,n]+N.\label{eq:interface}
\end{align}
\end{lemma}
\begin{IEEEproof}
The vertex sets of the two smaller bands are disjoint. Their optimal path families are valid in the larger band and can be combined, proving~\eqref{eq:super}.

For the reverse inequality, every arc crossing the transverse interface between rows $2k-1$ and $2k$ is diagonal and incident to one of the $N$ waypoints
\begin{equation}
B_k=\{z_{2h+1,2k-1}:h=0,\ldots,N-1\}.\label{eq:interfacevertices}
\end{equation}
Take an optimal vertex-disjoint family in the larger band. At most $N$ paths visit $B_k$. Discard those paths. Every remaining path stays entirely in one smaller band, so their number is at most $K_N[m,k]+K_N[k,n]$. This proves~\eqref{eq:interface}.
\end{IEEEproof}
The first inequality is exact superadditivity, without needing any boundary correction or independence assumption. The second explicitly considers the boundary effect: its cost is proportional to depth and need not vanish after width normalization when depth and width grow at the same rate.

\subsection{Existence and Finite-Width Approximation}
Write $K_{N,n}=K_N[0,n]$, $J_{N,n}=J_N[0,n]$, and $W_n=2bn$.
\begin{theorem}[current density in a Poisson forest belt]\label{thm:transverse}
For every fixed $N$, speed $\nu$, and parameters satisfying~\eqref{eq:spacing}, there is a deterministic constant $\kappa_N^\nu\in[0,1]$ such that
\begin{equation}
\frac{K_{N,n}}n\longrightarrow\kappa_N^\nu,\quad
\frac{J_{N,n}}{W_n}\longrightarrow j_N^\nu:=\frac{\kappa_N^\nu}{2bH}\label{eq:transverse}
\end{equation}
almost surely and in $L^1$ as $n\to\infty$. Moreover,
\begin{equation}
\kappa_N^\nu=\sup_{n\ge1}\frac{\E K_{N,n}}n,\label{eq:supmean}
\end{equation}
and, for every finite $n$,
\begin{equation}
0\le j_N^\nu-\frac{\E J_{N,n}}{W_n}\le\frac{N}{HW_n}.\label{eq:meanerror}
\end{equation}
\end{theorem}
\begin{IEEEproof}
Let $\theta$ shift the forest by one transverse period $2b$, with the sign chosen so that
\begin{equation}
K_N[m,m+n](\omega)=K_{N,n}(\theta^m\omega).\label{eq:stationarity}
\end{equation}
The Poisson process is stationary and ergodic under $\theta$. With $X_n=-K_{N,n}$,~\eqref{eq:super} gives
\begin{equation}
X_{m+n}\le X_m+X_n\circ\theta^m,\quad -n\le X_n\le0.\label{eq:kingman}
\end{equation}
Kingman's subadditive ergodic theorem~\cite{kingman1973,steele1989} gives an almost-sure deterministic limit for $X_n/n$. The deterministic bound gives $L^1$ convergence as well. The expectations $\E K_{N,n}$ are superadditive, so their normalized limit is their supremum, proving~\eqref{eq:supmean}. Equation~\eqref{eq:exactcurrent} gives the second limit in~\eqref{eq:transverse}.

To obtain the upper approximation bound, partition a width of $pn$ periods into $p$ bands of $n$ periods each and repeatedly apply~\eqref{eq:interface}. Stationarity gives
\begin{equation}
\E K_{N,pn}\le p\E K_{N,n}+(p-1)N.\label{eq:partitionmean}
\end{equation}
Divide by $pn$ and let $p\to\infty$. Together with~\eqref{eq:supmean}, this yields
\begin{equation}
0\le\kappa_N^\nu-\frac{\E K_{N,n}}n\le\tfrac{N}{n}.\label{eq:kappaerror}
\end{equation}
Divide by $2bH$ to obtain~\eqref{eq:meanerror}.
\end{IEEEproof}
The probabilistic input is transverse ergodicity. Dependence among neighboring corridor states causes no problem because the proof uses stationarity and ergodicity of the underlying forest, not independence of those states. The finite-width bound concerns the ensemble mean.

\subsection{Concentration About the Mean}
The same locality that produces dependence between nearby corridors limits the effect of a local change in the forest.
\begin{proposition}[finite-width concentration]\label{prop:concentration}
For fixed corridor geometry and clearance radius, there is a constant $c>0$, independent of $N,n$, and $\zeta$, such that
\begin{equation}
\Prob\{|K_{N,n}-\E K_{N,n}|\ge t\}
\le2\exp\left\{-\frac{ct^2}{(N+1)(n+1)}\right\}\label{eq:concentration}
\end{equation}
for $t>0$. Consequently, for $\varepsilon>0$,
\begin{align}
&\Prob\left\{\left|\frac{K_{N,n}}n-\kappa_N^\nu\right|>\tfrac{N}{n}+\varepsilon\right\}\notag\\
&\hspace{12mm}\le2\exp\left\{-\frac{c\varepsilon^2n^2}{(N+1)(n+1)}\right\}.\label{eq:concentrationlimit}
\end{align}
\end{proposition}
\begin{IEEEproof}
Partition the plane into half-open squares of a fixed positive side length. Only squares meeting the radius-$R$ neighborhoods of candidate corridors can affect $K_{N,n}$. There are at most $C_0(N+1)(n+1)$ such squares, where $C_0$ depends only on the fixed geometry and clearance. The restrictions of the Poisson process to these squares are independent.

Resampling the process in one square can change the states of at most $A_0$ candidate arcs, for a fixed finite $A_0$. From an optimal vertex-disjoint family before the change, discard every path using an affected arc. Each affected arc belongs to at most one selected path, so at most $A_0$ paths are lost. The same argument with the two configurations reversed proves that $K_{N,n}$ changes by at most $A_0$.

The bounded-differences inequality~\cite{mcdiarmid1989}, obtained by applying Hoeffding's exponential bound~\cite{hoeffding1963} to the associated martingale differences, proves~\eqref{eq:concentration}, with $c=2/(C_0A_0^2)$. Equation~\eqref{eq:kappaerror} then gives~\eqref{eq:concentrationlimit}.
\end{IEEEproof}
At fixed depth, the bound controls random fluctuations of $K_{N,n}/n$ on the $n^{-1/2}$ scale, whereas its mean bias is at most $N/n$. Note that bound~\eqref{eq:concentration} alone does not prove a joint-growth law when $N$ is proportional to $n$.

\section{Positivity, Depth, and Speed}\label{sec:depth}
\subsection{A Positive Current at Every Finite Depth}
Let $p_N$ be the probability that a fixed straight centerline of length $L=2Nd$ has clearance $R$. Its forbidden region for tree centers is a capsule of area $2RL+\pi R^2$, giving
\begin{equation}
p_N=\exp\{-\zeta(2RL+\pi R^2)\}.\label{eq:straightprob}
\end{equation}
\begin{proposition}[Poisson bounds and a solvable regime]\label{prop:finitepositive}
The current density satisfies
\begin{equation}
\frac{e^{-\zeta(2RL+\pi R^2)}}{2bH}\le j_N^\nu
\le\frac{e^{-\zeta\pi R^2}}{2bH}.\label{eq:finitepositive}
\end{equation}
In particular, it is strictly positive at every fixed finite depth. If $\nu b/d>u$, so that diagonal arcs are unavailable, then
\begin{equation}
j_N^\nu=\frac{e^{-\zeta(2RL+\pi R^2)}}{2bH}.\label{eq:straightonly}
\end{equation}
\end{proposition}
\begin{IEEEproof}
For each inlet row $j=2k$, consider the straight horizontal path across the entire belt. It is open if and only if its length-$L$ capsule contains no tree center. Paths in different inlet rows are vertex-disjoint. If $I_k$ denotes this void event, then
\begin{equation}
K_{N,n}\ge\sum_{k=0}^{n-1}\ind_{I_k}.\label{eq:straightlower}
\end{equation}
Each indicator has mean $p_N$. Taking expectations and using~\eqref{eq:supmean} gives $\kappa_N^\nu\ge p_N$. For the upper bound, every crossing needs a clear inlet waypoint. The probability that its radius-$R$ disk contains no tree center is $e^{-\zeta\pi R^2}$. Count these inlets and take expectations and then width limits.

When diagonal arcs are absent, no path changes its row, so every crossing is one of the straight paths in~\eqref{eq:straightlower} and equality holds. Alternatively, the transverse ergodic theorem gives the same equality directly. Independence between neighboring row indicators is not needed.
\end{IEEEproof}
This explicit example separates finite-depth transport from indefinite flight. In the straight-only regime, every fixed row is almost surely blocked somewhere along an infinite forward ray, but a positive fraction of rows cross any specified finite depth. The fraction decreases exponentially with depth. Thus positive finite-belt current neither implies an infinite trajectory nor contradicts a single-bird critical-speed result.

\subsection{Increasing the Forest Depth}
\begin{proposition}[depth monotonicity]\label{prop:depth}
For a common forest realization and fixed width,
\begin{equation}
K_{N+1,n}\le K_{N,n}.\label{eq:depthmono}
\end{equation}
Consequently $j_{N+1}^\nu\le j_N^\nu$, and the iterated limits
\begin{equation}
j_\infty^\nu=\lim_{N\to\infty}j_N^\nu=\inf_{N\ge1}j_N^\nu,\quad
\kappa_\infty^\nu=2bHj_\infty^\nu\label{eq:iterated}
\end{equation}
exist.
\end{proposition}
\begin{IEEEproof}
Every path to layer $2(N+1)$ visits layer $2N$: no allowed arc skips an even layer. Truncating a vertex-disjoint family at that layer gives a vertex-disjoint crossing family for the shorter belt. Taking expectations and then width limits proves the assertions.
\end{IEEEproof}
The order in~\eqref{eq:iterated} is width first and depth second. Section~\ref{sec:bulk} identifies joint-growth sequences with the same limit. The following bound does not restrict their relative growth rates.
\begin{proposition}[survival bound]\label{prop:survival}
Let $p_\infty^\nu$ be the probability that a specified inlet waypoint has an infinite open forward path in the full corridor graph. For any integer sequences $N_k\to\infty$ and $n_k\to\infty$,
\begin{equation}
\limsup_{k\to\infty}\frac{J_{N_k,n_k}}{2bn_k}
\le\frac{p_\infty^\nu}{2bH}\quad\text{almost surely}.\label{eq:survivalbound}
\end{equation}
If $p_\infty^\nu=0$, the normalized current converges to zero almost surely and in $L^1$. Also $\kappa_\infty^\nu\le p_\infty^\nu$.
\end{proposition}
\begin{IEEEproof}
Fix an integer $M$. Let $A_M(v)$ be the event that an inlet waypoint $v$ can reach longitudinal layer $2M$ in the full graph, allowing arbitrary transverse excursions. For $N_k\ge M$, distinct crossing paths start at distinct inlet waypoints for which $A_M(v)$ occurs. Therefore
\begin{equation}
K_{N_k,n_k}\le\sum_{h=0}^{n_k-1}\ind_{A_M(z_{0,2h})}.\label{eq:survivalcount}
\end{equation}
By transverse ergodicity, the right side divided by $n_k$ converges almost surely to $\Prob(A_M(z_{0,0}))$. This holds on a common probability-one event for all positive integers $M$. The events decrease with $M$, and their intersection is the event of an infinite open path: finite outdegree and strict forward progress give this implication by the finitely branching infinite-path principle. Letting $M\to\infty$ proves~\eqref{eq:survivalbound}. Nonnegativity and the deterministic bound $J_{N_k,n_k}/(2bn_k)\le1/(2bH)$ give $L^1$ convergence when $p_\infty^\nu=0$. Applying the same inlet count at fixed depth $M$, taking width limits, and then letting $M\to\infty$ proves $\kappa_\infty^\nu\le p_\infty^\nu$.
\end{IEEEproof}
An infinite open corridor path would concatenate to a collision-free trajectory of~\eqref{eq:dynamics}, with clearance for the same bird radius. Thus absence of infinite trajectories in that physical model also forces $p_\infty^\nu=0$.

\subsection{Separating Service Rate from Maneuverability}
For speed comparisons, let the headway be
\begin{equation}
H(\nu)=s/\nu,\quad s>2q,\label{eq:speedheadway}
\end{equation}
so $s$ is the required longitudinal separation. On a common forest and fixed corridor geometry, the available arc set decreases as $\nu$ increases. Hence $\kappa_N^\nu$ is nonincreasing and
\begin{equation}
\frac{j_N^{\nu_2}}{\nu_2}\le\frac{j_N^{\nu_1}}{\nu_1},\quad \nu_2>\nu_1.\label{eq:speedmono}
\end{equation}
The same inequality holds for $j_\infty^\nu$. The service-normalized current is therefore nonincreasing, even though current itself need not be. This coarse library has only one lateral-maneuver threshold. In particular, its straight-only finite-depth current grows proportionally to $\nu$ under this headway convention, while its infinite-depth current is zero.

\section{Admission Control and Congestion}\label{sec:queue}
\subsection{Offline Route Construction}
The throughput-achieving controller first determines which candidate corridors are open, builds the vertex-split network, computes an integral maximum flow, and extracts its $K=K_{N,n}$ paths. It then launches at most one bird per selected path at each epoch $kH$. All forest transit takes the fixed time $L/\nu$. For deterministic saturated demand, Theorem~\ref{thm:finite} achieves the capacity exactly and maximizes deliveries at every horizon. Thinning the launch sequence gives any smaller long-run rate. With fixed selected paths, no intermediate queue or backpressure optimization is necessary. This is due to the matched corridor, headway, and capacity model. Offline geometry and online admission are separated. The known forest determines an integral maximum flow and a fixed family of $K$ disjoint routes. At each launch epoch, the upstream queue sends at most one bird on each route. All admitted birds cross in time $L/\nu$, with no queue inside the forest.

\subsection{Stochastic Arrivals}
Now consider stochastic arrivals. Fix a realized finite graph with $K>0$. Suppose a staging facility can dispatch to all $K$ selected inlet waypoints. Let $Q_k$ be its queue just before the $k$th launch epoch. Let $A_k$ be the number of arrivals after that epoch and before the next. Assume the $A_k$ are i.i.d., independent of $Q_0$, with
\begin{equation}
\E A_k=\alpha<K,\quad \E A_k^2<\infty,\quad \E Q_0^2<\infty.\label{eq:arrival}
\end{equation}
Launch $\min\{Q_k,K\}$ birds along distinct selected paths. Then
\begin{equation}
Q_{k+1}=[Q_k-K]^++A_k.\label{eq:queue}
\end{equation}
The realized forest and path family remain fixed during operation.
\begin{theorem}[stable admission below capacity]\label{thm:queue}
The admission policy strongly stabilizes~\eqref{eq:queue}, with
\begin{equation}
\limsup_{M\to\infty}\tfrac{1}{M}\sum_{k=0}^{M-1}\E Q_k
\le\frac{\E A_0^2+K^2}{2(K-\alpha)}.\label{eq:backlog}
\end{equation}
Mean cumulative departures, divided by the number of launch epochs, converge to $\alpha$. All admitted birds traverse the forest without waiting, and mean forest output has the same long-run rate after the fixed transit delay.
\end{theorem}
\begin{IEEEproof}
With $V(Q)=Q^2/2$, the inequality
\begin{equation}
([Q-K]^++A)^2\le Q^2+K^2+A^2+2Q(A-K)\label{eq:queueineq}
\end{equation}
yields the conditional drift bound
\begin{align}
\E[V(Q_{k+1})-V(Q_k)\mid Q_k]&\le B-(K-\alpha)Q_k,\notag\\
B&=\tfrac12(\E A_0^2+K^2).\label{eq:drift}
\end{align}
Taking expectations, summing, and using $V\ge0$ proves~\eqref{eq:backlog}. Dropping the negative term instead gives $\E Q_M^2\le\E Q_0^2+2BM$, so $\E Q_M/M\to0$. Summing the exact queue equation therefore shows that mean cumulative launches divided by $M$ converge to $\alpha$. Each launch arrives at the outlet after $L/\nu$, so the same long-run mean rate holds there. Collision avoidance follows because the schedule is a thinning of the schedule in Theorem~\ref{thm:finite}.
\end{IEEEproof}
The drift argument is a standard queueing capacity mechanism~\cite{neely2010}, here with buffering outside the forest. It guarantees finite time-average expected backlog. Fixed arrivals at separate inlet locations need not be dispatchable in this way and are not covered by a condition on aggregate rate alone.

\subsection{Bulk Load Condition}
Suppose input traffic has mean rate $\lambda W_n$ birds per unit time. The mean number arriving per launch interval is $\alpha_n=\lambda W_nH$. For any fixed $\lambda<j_N^\nu$, Theorem~\ref{thm:transverse} gives
\begin{equation}
\lambda W_nH<K_{N,n}\label{eq:load}
\end{equation}
for every sufficiently large $n$, almost surely. Thus each sufficiently wide realized belt admits the stable upstream policy under the stated moment and dispatch assumptions.

If $\lambda>j_N^\nu$, the cut bound is smaller than the offered rate for every sufficiently large width, almost surely. More explicitly, for any finite realization and an arrival process satisfying $A(T)/T\to\varrho>K/H$, conservation and Theorem~\ref{thm:finite} imply
\begin{equation}
\liminf_{T\to\infty}\frac{A(T)-D_\pi(T)}T\ge\varrho-\tfrac{K}{H}>0.\label{eq:overload}
\end{equation}
Without rejection, accumulated demand grows at least linearly. Since every admitted bird spends exactly $L/\nu$ inside the forest, and inlet headways bound the number admitted during any such interval, forest inventory is uniformly bounded for a fixed finite graph. Thus the growing backlog must remain upstream. The cut capacity is an operational congestion threshold for this protocol. This conservation interpretation parallels capacity-limited traffic~\cite{lighthill1955,richards1956}. A saturated internal cut indicates there is an obstruction even though admission keeps the physical backlog outside the forest.

At exact critical loading, deterministic arrivals can be matched to the periodic schedule, but the stochastic stability theorem requires strict slack. The distinction between capacity and delay is important: exact deterministic throughput does not imply uniformly small delay for stochastic demand.

\section{Two-Dimensional Bulk Transport}\label{sec:bulk}
\subsection{Another Scaling for the Flight Problem}
The fixed-depth law uses an exact path-combination inequality. Its accompanying interface estimate costs $N$ waypoints, proportional to depth. When $L=2Nd$ and $W=2bn$ grow proportionally, this estimate no longer gives a vanishing error after normalization by $W$. We therefore use a different probabilistic route for this regime, while still using the same definitions for forest, corridor library, waypoint headway, and no-stopping schedule from Sections~\ref{sec:model}--\ref{sec:queue}.

Max-flow/min-cut duality reduces the bulk-transport problem to the capacity of separating cuts in the random corridor network. For i.i.d. lattice capacities, Rossignol and Th{\'e}ret establish laws of large numbers for tilted cylinders~\cite{rossignol2010tilted} and related box laws and lower large deviations~\cite{rossignol2010box}. These results also provide positivity information under specified assumptions. Forest corridors, however, are directed and dependent because neighboring swept tubes have overlapping clearance neighborhoods. We prove the required asymmetric dual shape theorem directly, using the deterministic geometry to control endpoint perturbations.

\subsection{A Bounded Planar Representation of Waypoint Capacity}
Fix $\nu$ and suppress it in the auxiliary notation. Consider the full geometric library from before, including every diagonal candidate, and assign a closed or dynamically unavailable corridor capacity zero. Replace each waypoint $v$ by $v^-,v^+$ and set
\begin{align}
c(v^-\to v^+)&=1,\notag\\
c(v^+\to w^-)&=\ind\{v\to w\text{ is available and open}\}.\label{eq:boundedcaps}
\end{align}
In a finite band, the source set consists of the inlet $v^-$ vertices and the sink set of the outlet $v^+$ vertices, with unrestricted terminal access. All capacities in the spatial network are now bounded independently of the band width. Only a purely computational supersource or supersink would require additional attachment arcs.

\begin{lemma}[flow-equivalent planar realization]\label{lem:planar}
The maximum flow in~\eqref{eq:boundedcaps} is $K_{N,n}$. The infinite auxiliary graph has a periodic planar embedding. It can be completed by deterministic zero-capacity arcs so that its cells have uniformly bounded diameter, without changing any finite-band flow.
\end{lemma}
\begin{IEEEproof}
Every corridor flow is already at most one: it leaves a waypoint of total capacity one and enters another such waypoint. Replacing the larger corridor capacities of Section~\ref{sec:finite} by the capacities in~\eqref{eq:boundedcaps} therefore leaves the feasible flows unchanged. Retained zero-capacity arcs carry no flow.

Planar vertex splitting is not valid for arbitrary directed embeddings~\cite{kaplan2011}. Here all incoming arcs approach from the left and all outgoing arcs leave to the right; the two groups are consecutive in the cyclic order around each waypoint. Choose $0<\sigma<d/4$, place $v^\pm=v\pm\sigma e_1$, and join the corresponding split endpoints in the same noncrossing order. The horizontal arc $v^-\to v^+$ separates the two groups without introducing a crossing. This is an auxiliary drawing rather than a change to the flight centerlines used to define openness.

For an explicit cell completion, at each even longitudinal layer $i=2h$ join consecutive $v^-$ vertices vertically, and do the same for the $v^+$ vertices. Give these auxiliary arcs capacity zero and orient them upward. No corridor skips an even layer, so these vertical segments meet existing arcs only at endpoints. Keep them in every finite band whenever both endpoints are retained. The completed graph is periodic under
\begin{equation}
\Lambda=2d\Z\times2b\Z,\label{eq:lattice}
\end{equation}
and has finitely many bounded polygonal cell types. The added arcs have no effect on any feasible flow.
\end{IEEEproof}
The gate capacities are deterministic, and each random corridor capacity is determined by tree centers in a uniformly bounded neighborhood of its physical segment. Thus the capacity field of the completed graph is stationary and finite-range dependent under~\eqref{eq:lattice}. This conclusion holds even though adjacent corridor capacities are dependent.

\subsection{Asymmetric Duality and the Slab Boundary}
Let $f_L(e)$ and $f_R(e)$ be the faces to the left and right of a directed auxiliary arc $e$. Introduce both geometric dual crossing directions, with lengths
\begin{equation}
\ell^*(f_L(e),f_R(e))=c_e,\quad
\ell^*(f_R(e),f_L(e))=0.\label{eq:duallengths}
\end{equation}
An outgoing cut arc is charged its capacity, whereas an incoming cut arc is not charged. This is the nonnegative-capacity form of directed planar flow duality~\cite{hassin1981,erickson2010,zenklusen2010}. The zero-length reverse crossings are essential: a directed cut can contain incoming arcs as well as outgoing ones. An outgoing-capacity sum does not charge both.

Choose one dual vertex in each cell in a translation-equivariant manner. In the infinite completed dual, let
\begin{equation}
T^*(p,q)=\inf_{\pi:p\to q}\sum_{e^*\in\pi}\ell^*(e^*)\label{eq:passage}
\end{equation}
over finite dual paths. The bidirected geometry makes $T^*$ finite. Periodicity and the bounded cell types imply constants $c_g,r_g<\infty$ such that a straight segment can be followed within distance $r_g$ by a dual path of at most $c_g(|p-q|+1)$ arcs. Consequently,
\begin{equation}
0\le T^*(p,q)\le c_g(|p-q|+1).\label{eq:routing}
\end{equation}
Indeed, a slightly perturbed segment crosses only $O(|p-q|+1)$ cells, since the embedding has uniformly bounded local complexity. Visiting those cells in order gives the stated path. Passage time need not be symmetric or strictly positive.

For a finite band with $n\ge2$, the zero-capacity chains through the inlet $v^-$ vertices and the outlet $v^+$ vertices form the left and right boundary chains. Their vertices have fixed source and sink labels, respectively. Split the remaining exterior face into an upper and a lower port; do not permit dual exits through the terminal chains. Write $D_{N,n}$ for the dual graph of bounded faces, and $\mathcal A_{N,n},\mathcal B_{N,n}$ for its vertices adjacent to the upper and lower boundary arcs. Let $T^*_{N,n}$ denote passage time restricted to $D_{N,n}$.

\begin{lemma}[cut paths with bounded end corrections]\label{lem:boundary}
For the completed auxiliary graph,
\begin{equation}
0\le K_{N,n}-\min_{\substack{p\in\mathcal A_{N,n}\\q\in\mathcal B_{N,n}}}
T^*_{N,n}(p,q)\le2.\label{eq:boundary}
\end{equation}
Every bounded face of the finite graph is a face of the infinite completed graph. Its upper and lower attachment vertices are within a fixed distance of the corresponding horizontal sides of $[0,L]\times[0,W]$. Away from a fixed-width boundary collar, the finite and infinite duals coincide. Interior dual paths can be joined to nearby upper or lower attachments with bounded cost when their endpoints have bounded distance from those sides and stay away from the lateral sides.
\end{lemma}
\begin{IEEEproof}
All potential corridors, even zero-capacity ones, are retained. To describe the cells and boundary explicitly, let $y_k=b(2k+\tfrac12)$ and write
\begin{equation*}
\begin{aligned}
E_{h,k}^{\pm}&=(2hd\pm\sigma,y_k),\\
O_{h,k}^{\pm}&=((2h+1)d\pm\sigma,y_k+b)
\end{aligned}
\end{equation*}
for the split even- and odd-layer waypoints. Take the auxiliary corridors to be straight segments between their split endpoints. At each even layer, the completion arcs and two consecutive gates bound a rectangle with cyclic vertex list
\begin{equation*}
(E_{h,k}^-,E_{h,k}^+,E_{h,k+1}^+,E_{h,k+1}^-).
\end{equation*}
Between consecutive even layers, each region between heights $y_k$ and $y_{k+1}$ has two triangular faces and two trapezoidal faces. Their cyclic vertex lists are
\begin{equation*}
\begin{aligned}
&(E_{h,k}^+,O_{h,k}^-,E_{h,k+1}^+),\\
&(E_{h+1,k}^-,E_{h+1,k+1}^-,O_{h,k}^+),\\
&(E_{h,k}^+,E_{h+1,k}^-,O_{h,k}^+,O_{h,k}^-),\\
&(E_{h,k+1}^+,O_{h,k}^-,O_{h,k}^+,E_{h+1,k+1}^-).
\end{aligned}
\end{equation*}
In the finite band, retain the rectangles for $0\le h\le N$, $0\le k\le n-2$, and all four intervening faces for $0\le h<N$, $0\le k\le n-2$. In addition, retain the lower trapezoid, the third face in the list, at $k=n-1$ for each $0\le h<N$. These are all the bounded faces of the finite graph. Each is an unchanged face of the infinite completed graph.

The left terminal chain joins $E_{0,k}^-$ to $E_{0,k+1}^-$, and the right terminal chain joins $E_{N,k}^+$ to $E_{N,k+1}^+$, for $0\le k\le n-2$. The lower boundary alternates gates and straight corridors at height $y_0=b/2$. The upper boundary alternates even-layer gates with the segments
\begin{equation*}
E_{h,n-1}^+\longrightarrow O_{h,n-1}^-
\longrightarrow O_{h,n-1}^+\longrightarrow E_{h+1,n-1}^-,
\end{equation*}
so its height lies between $W-3b/2$ and $W-b/2$. The upper attachment faces are the rectangles at $k=n-2$ and the lower trapezoids at $k=n-1$; the lower attachment faces are the rectangles and lower trapezoids at $k=0$. Thus both attachment sets stay within a fixed distance of the horizontal sides. The upper and lower ports adjoin only the respective boundary arcs; crossings through the two terminal chains are excluded.

The cell list also identifies a fixed-width collar outside which the finite and infinite duals coincide. A vertical segment from an interior point at bounded distance from a horizontal side crosses only a bounded number of cells. Perturbing it to avoid primal vertices, and keeping it away from the terminal chains, gives a bounded-cost route to the corresponding attachment set.

Every source--sink cut has a boundary component running from the upper to the lower port. Traverse it with the source side on the right. Outgoing primal crossings have length $c_e$ under~\eqref{eq:duallengths}, and incoming crossings have length zero. Removing other boundary components and erasing loops cannot increase cost. Conversely, a simple port-to-port dual path separates the two terminal chains and charges exactly the outgoing capacity of its associated cut. Thus $K_{N,n}$ equals the minimum port-to-port passage cost. Deleting the first and last dual crossing leaves a path in $D_{N,n}$. Each deleted arc has length at most one. Conversely, any path between the stated attachments extends to a port-to-port path by two such arcs. This proves~\eqref{eq:boundary}.
\end{IEEEproof}
The end correction concerns two arcs, not a modification of an entire lateral boundary. The explicit terminal chains exclude an artificial exterior bypass. These details allow the infinite passage model to be used without changing the finite flight capacity.

\subsection{A Uniform Asymmetric Shape Theorem}
Stationary first-passage percolation and random semimetrics provide the relevant background~\cite{boivin1990,haggstrom1995,bjorklund2010}. Here the deterministic routing bound gives a direct proof, including moving endpoints. In particular, the argument does not require a positive time constant.

\begin{theorem}[uniform asymmetric dual shape]\label{thm:shape}
For the completed auxiliary graph, there is a deterministic, finite, nonnegative, Lipschitz continuous, convex, positively homogeneous function $\mu^\nu:\R^2\to[0,\infty)$ such that, on a common probability-one event, for every compact $C\subset\R^2$,
\begin{equation}
\lim_{r\to\infty}\tfrac{1}{r}
\sup_{p,q\in V^*\cap rC}
\bigl|T^*(p,q)-\mu^\nu(q-p)\bigr|=0.\label{eq:shape}
\end{equation}
Here $V^*$ is the infinite dual vertex set. The function may be asymmetric or vanish in some directions.
\end{theorem}
\begin{IEEEproof}
Suppress $\nu$ in $\mu$. We first record the endpoint estimate. The directed triangle inequality and~\eqref{eq:routing}, applied in both orders, give
\begin{align}
|T^*(p,q)-T^*(p',q')|
\le c_g\bigl(|p-p'|+|q-q'|+2\bigr).\label{eq:endpoints}
\end{align}
Thus changing either endpoint by a small fraction of the observation scale changes passage time by at most the same order.

Fix a dual vertex $o$. For each nonzero $z\in\Lambda$, the array
\begin{equation}
X_{m,n}=T^*(o+mz,o+nz),\quad m<n,\label{eq:directionarray}
\end{equation}
is stationary and subadditive. Translation by $z$ is ergodic for the Poisson process, and~\eqref{eq:routing} gives integrability. Kingman's theorem therefore yields a deterministic limit
\begin{equation}
g(z)=\lim_{k\to\infty}\frac{T^*(o,o+kz)}k\label{eq:directionlimit}
\end{equation}
almost surely and in $L^1$. Set $g(0)=0$. Countability gives these limits simultaneously for all $z\in\Lambda$. Taking subsequences proves $g(kz)=kg(z)$ for positive integers $k$. Moreover,
\begin{align}
\E T^*(o,o+k(z+w))
&\le\E T^*(o,o+kz)\notag\\
&\quad+\E T^*(o+kz,o+k(z+w))\notag\\
&=\E T^*(o,o+kz)+\E T^*(o,o+kw).
\label{eq:limittriangle}
\end{align}
Divide by $k$ and pass to the limit to obtain $g(z+w)\le g(z)+g(w)$. Since $0\le g(z)\le c_g|z|$, applying subadditivity in both orders gives
\begin{equation}
|g(z)-g(w)|\le c_g|z-w|.\label{eq:glipschitz}
\end{equation}
Integer homogeneity extends $g$ consistently to rational multiples of $\Lambda$. Equation~\eqref{eq:glipschitz} then gives a unique Lipschitz extension $\mu$ to $\R^2$. It is nonnegative, subadditive, and positively homogeneous, hence convex.

Next, directional convergence is uniform over distant endpoints from $o$. For any angular tolerance, choose finitely many nonzero lattice vectors whose directions form a net of the unit circle. A sufficiently distant point in $o+\Lambda$ can be approximated by a positive integer multiple of one of these vectors, with error at most the angular tolerance times its distance from $o$, plus a constant depending on the net. Apply~\eqref{eq:directionlimit} on these finitely many rays, and use~\eqref{eq:endpoints} and the Lipschitz bound for $\mu$ to control the approximation. The periodic dual has finitely many vertex types, so every $q\in V^*$ is within a bounded distance of $o+\Lambda$. Another use of~\eqref{eq:endpoints} yields
\begin{equation}
\lim_{M\to\infty}\sup_{\substack{q\in V^*\\|q-o|\ge M}}
\frac{|T^*(o,q)-\mu(q-o)|}{|q-o|}=0
\quad\text{almost surely}.\label{eq:fixedorigin}
\end{equation}

It remains to make the origin move. Fix $\varepsilon>0$. An origin $x\in o+\Lambda$ is called good at scale $M$ if
\begin{equation}
|T^*(x,q)-\mu(q-x)|\le\varepsilon|q-x|
\quad\text{for all }q\in V^*,\ |q-x|\ge M.\label{eq:good}
\end{equation}
By~\eqref{eq:fixedorigin}, the probability $\beta_M$ that $o$ is not good tends to zero as $M\to\infty$.

Fix a compact set $C$, a box $Q$ containing its unit neighborhood, and $0<\eta_0<1$. Every ball of radius $\eta_0r$ centered in $rC$ contains at least $a_1\eta_0^2r^2$ points of $o+\Lambda$ for all sufficiently large $r$, where $a_1>0$ depends only on the lattice. Choose $M$ so large that
\begin{equation}
\beta_M\frac{|Q|}{4db}<\frac{a_1\eta_0^2}{2}.\label{eq:baddensitychoice}
\end{equation}
The spatial ergodic theorem for the lattice action~\cite{krengel1985} implies that the number of bad origins in $rQ$, divided by $r^2$, converges almost surely to $\beta_M|Q|/(4db)$. Thus, eventually, no radius-$\eta_0r$ ball centered in $rC$ can consist entirely of bad origins. Every $p\in V^*\cap rC$ is within $\eta_0r$ of a good origin $x$.

For $p,q\in V^*\cap rC$ with $|q-x|\ge M$,~\eqref{eq:endpoints} and~\eqref{eq:good} give
\begin{align}
|T^*(p,q)-\mu(q-p)|
&\le2c_g\eta_0r+2c_g\notag\\
&\quad+\varepsilon(\diam(C)+\eta_0)r.\label{eq:movingbound}
\end{align}
If $|q-x|<M$, use~\eqref{eq:routing} instead to bound the same expression by $2c_g\eta_0r+O(M+1)$. Divide by $r$ and let $r\to\infty$, then let $\eta_0$ and $\varepsilon$ decrease to zero. A countable choice of tolerances and enclosing boxes gives~\eqref{eq:shape} on one probability-one event for all compact sets.
\end{IEEEproof}
The spatial ergodic theorem controls the density of good origins without requiring independence of their defining events, which depend on the entire forest. Together with the deterministic endpoint estimate, this density bound extends the passage-time approximation from good origins to nearby points.

\subsection{Reflection Symmetry and the Vertical Direction}
The corridor library distinguishes forward from backward motion, but treats the two lateral directions alike. In the dual this gives horizontal evenness, rather than full symmetry.
\begin{lemma}[horizontal evenness]\label{lem:symmetry}
The dual time constant satisfies
\begin{equation}
\mu^\nu(x,y)=\mu^\nu(-x,y).\label{eq:symmetry}
\end{equation}
Consequently, for every real $\xi$,
\begin{equation}
\mu^\nu(\xi,-1)\ge\mu^\nu(0,-1).\label{eq:verticalmin}
\end{equation}
\end{lemma}
\begin{IEEEproof}
The reflection $S(x,y)=(x,b-y)$ sends $z_{i,j}$ to $z_{i,-j}$ and preserves the corridor library and the Poisson law. It also preserves the split-vertex construction. It reverses the orientation of the vertical completion arcs, but both dual crossing directions of those zero-capacity arcs have length zero.

Reflection interchanges the left and right faces of each primal corridor or gate. Therefore, reflecting and reversing a dual path preserves its passage cost in the reflected forest. With $S$ also denoting the induced map on dual faces,
\begin{equation}
T^*_{S\omega}(Sq,Sp)=T^*_\omega(p,q).\label{eq:reflectpassage}
\end{equation}
Dual representatives can be chosen compatibly with reflection; any bounded adjustment of representatives has no effect on the time constant. If $q-p=(x,y)$, the displacement from $Sq$ to $Sp$ is $(-x,y)$. Apply Theorem~\ref{thm:shape} and use the deterministic nature of $\mu^\nu$ to obtain~\eqref{eq:symmetry}. Convexity now gives
\begin{equation*}
\mu^\nu(0,-1)\le\tfrac12\bigl(\mu^\nu(\xi,-1)+\mu^\nu(-\xi,-1)\bigr)
=\mu^\nu(\xi,-1).\tag*{\IEEEQEDhere}
\end{equation*}
\end{IEEEproof}

\subsection{Confined Vertical Blocks and the Iterated Limit}
An unrestricted near-optimal dual path can leave a finite slab. The following argument provides confinement instead of presuming it. Only vertical blocks are needed because of Lemma~\ref{lem:symmetry}.
\begin{proposition}[identification of the iterated current]\label{prop:identify}
For every fixed speed and admissible geometry,
\begin{equation}
j_\infty^\nu=\frac{\mu^\nu(0,-1)}H.\label{eq:identification}
\end{equation}
\end{proposition}
\begin{IEEEproof}
Let $c_0=\mu^\nu(0,-1)$. First fix a depth $N$. The horizontal separation of any upper and lower attachments in Lemma~\ref{lem:boundary} is bounded by a constant depending on $N$, and their vertical separation is $-W_n+O(1)$. Restricted passage cost is no smaller than unrestricted passage cost. Applying Theorem~\ref{thm:shape} to all these endpoints gives
\begin{equation}
\liminf_{n\to\infty}\frac{K_{N,n}}{W_n}\ge c_0.\label{eq:stripshapelower}
\end{equation}
Theorem~\ref{thm:transverse} implies $Hj_N^\nu\ge c_0$. Taking the infimum over $N$ proves $Hj_\infty^\nu\ge c_0$.

For the reverse inequality, let $z=(0,-2b)\in\Lambda$ and fix a dual vertex $o$. For an integer $m\ge1$, let $X_{m,B}$ be the passage time from $o$ to $o+mz$, restricted to dual vertices in the ball of radius $B$ about $o$. For all sufficiently large $B$, the geometric routing bound gives
\begin{equation}
0\le X_{m,B}\le c_g(m|z|+1).\label{eq:blockbound}
\end{equation}
As $B\to\infty$, $X_{m,B}$ decreases to $T^*(o,o+mz)$, since every finite path lies in some such ball. Directional $L^1$ convergence gives
\begin{equation}
\lim_{m\to\infty}\frac{\E T^*(o,o+mz)}{2bm}=c_0.\label{eq:verticalmean}
\end{equation}
Given $\varepsilon>0$, first choose $m$, then $B$, so that
\begin{equation}
\frac{\E X_{m,B}}{2bm}\le c_0+\varepsilon.\label{eq:blockapprox}
\end{equation}
Both remain fixed as the slab width grows.

Choose a geometric margin $\rho$ larger than the sum of the boundary-collar width in Lemma~\ref{lem:boundary}, the routing radius $r_g$, and the maximum dual-edge extent. Fix one depth $N_0$ large enough to contain a vertical chain of radius-$B$ blocks at distance at least $\rho$ from the lateral sides. For each sufficiently large $n$, choose the top block center $p_n\in o+\Lambda$ so that
\begin{equation*}
B+\rho<W_n-(p_n)_2\le B+\rho+2b.
\end{equation*}
Keep its horizontal coordinate fixed and stack translates in direction $mz$. Let $M_n$ be the largest integer for which the final endpoint $q_n=p_n+M_nmz$ has height at least $B+\rho$. Then
\begin{equation*}
B+\rho\le(q_n)_2<B+\rho+2bm.
\end{equation*}
Every block ball therefore stays at distance at least $\rho$ from the horizontal sides as well as the lateral sides. The two endpoint distances from the horizontal sides are bounded independently of $n$. There are
\begin{equation}
M_n=\tfrac{W_n}{2bm}+O(1)\label{eq:blockcount}
\end{equation}
blocks, all contained in the finite dual. Lemma~\ref{lem:boundary} provides bounded-cost end connections.

Let $X_i^{(n)}$ be the translated restricted costs. Concatenating their paths and the end connections yields
\begin{equation}
K_{N_0,n}\le\sum_{i=1}^{M_n}X_i^{(n)}+C_{m,B},\label{eq:blockupper}
\end{equation}
where $C_{m,B}$ is independent of $n$. Stationarity gives $\E X_i^{(n)}=\E X_{m,B}$. Independence is unnecessary. Take expectations, divide by $W_n$, and use Theorem~\ref{thm:transverse} to obtain
\begin{equation}
Hj_{N_0}^\nu\le\frac{\E X_{m,B}}{2bm}\le c_0+\varepsilon.\label{eq:stripshapeupper}
\end{equation}
Since $j_\infty^\nu\le j_{N_0}^\nu$, letting $\varepsilon\downarrow0$ proves~\eqref{eq:identification}.
\end{IEEEproof}
The depth $N_0$ depends on the prescribed accuracy but is fixed before the width limit. Thus the proof identifies the iterated limit without interchanging limits.

\subsection{The Joint-Growth Current Law}
\begin{theorem}[bulk current]\label{thm:bulk}
Let $N_n$ be deterministic positive integers with $N_n\to\infty$ and $N_n=O(n)$. Under the Poisson corridor protocol,
\begin{equation}
\frac{J_{N_n,n}}{2bn}\longrightarrow j_\infty^\nu
=\frac{\mu^\nu(0,-1)}H\label{eq:bulk}
\end{equation}
almost surely and in $L^1$. In particular, whenever $L_n/W_n\to\gamma\in(0,\infty)$,
\begin{equation}
j_{\rm bulk}^\nu(\gamma)
=\tfrac{1}{H}\min_{|\xi|\le\gamma}\mu^\nu(\xi,-1)
=j_\infty^\nu.\label{eq:aspect}
\end{equation}
The repeated-path no-stopping policy attains the finite-slab currents on the left of~\eqref{eq:bulk}.
\end{theorem}
\begin{IEEEproof}
For every fixed $M$, eventually $N_n\ge M$. Depth monotonicity gives $J_{N_n,n}\le J_{M,n}$. Divide by $W_n$, use Theorem~\ref{thm:transverse}, and then let $M\to\infty$ to obtain the upper limit $j_\infty^\nu$.

For the lower limit, choose minimizing attachments $p_n,q_n$ in~\eqref{eq:boundary}. Since $N_n=O(n)$, their locations divided by $W_n$ remain in a fixed compact set. The boundary lemma gives
\begin{equation}
\frac{q_n-p_n}{W_n}=(\xi_n,-1)+o(1),\quad \sup_n|\xi_n|<\infty.\label{eq:endpointscaled}
\end{equation}
Restricted passage cost is no smaller than unrestricted passage cost. Theorem~\ref{thm:shape}, positive homogeneity, and Lemma~\ref{lem:symmetry} therefore imply
\begin{equation}
\frac{K_{N_n,n}}{W_n}
\ge\mu^\nu(\xi_n,-1)-o(1)
\ge\mu^\nu(0,-1)-o(1).\label{eq:bulklower}
\end{equation}
Use Proposition~\ref{prop:identify} and $J_{N_n,n}=K_{N_n,n}/H$ to match the upper limit. Boundedness by $1/(2bH)$ gives $L^1$ convergence. Finally,~\eqref{eq:verticalmin} proves the minimization identity in~\eqref{eq:aspect}.
\end{IEEEproof}
The aspect-ratio independence uses the lateral symmetry of this corridor library. The theorem also covers sublinear diverging depths and bounded oscillating aspect ratios. All coordinates in the shape theorem are physical coordinates; the factor $1/H$ converts path-count capacity to birds per unit time.

Combining Theorem~\ref{thm:bulk} with Proposition~\ref{prop:survival} gives
\begin{equation}
0\le j_{\rm bulk}^\nu(\gamma)=j_\infty^\nu
\le\frac{p_\infty^\nu}{2bH}.\label{eq:comparison}
\end{equation}
In the straight-only regime,~\eqref{eq:straightonly} gives $j_\infty^\nu=0$, and Proposition~\ref{prop:survival} already gives zero current for arbitrary jointly growing dimensions.

The operational distinction between capacity and delay persists. If $\lambda<j_\infty^\nu$, the arrival rate $\lambda W_n$ is below $K_{N_n,n}/H$ for every sufficiently large slab in Theorem~\ref{thm:bulk}, almost surely. Under the arrival-moment and dispatch assumptions of Section~\ref{sec:queue}, its upstream admission queue is stable. Rates above $j_\infty^\nu$ overload a cut. Each statement takes long-run operation in a fixed finite slab before the spatial limit. No delay bound uniform in slab size is implied; the deterministic flight time $L_n/\nu$ itself grows with depth.

\subsection{Fixed Width and Reversed Limits}
The equality in~\eqref{eq:aspect} does not make the order of all spatial limits immaterial.
\begin{proposition}[eventual blocking at fixed width]\label{prop:fixedwidth}
Fix $\zeta>0$ and $n\ge1$. Then $K_{N,n}=0$ for all sufficiently large $N$, almost surely. More precisely, for any
\begin{equation}
0<r_0<\min\{R,b,d\},\quad
p_n=(1-e^{-\zeta\pi r_0^2})^n,\label{eq:blocklayerprob}
\end{equation}
one has
\begin{equation}
\Prob\{K_{N,n}>0\}\le(1-p_n)^N\le e^{-Np_n}.\label{eq:blocklayerbound}
\end{equation}
\end{proposition}
\begin{IEEEproof}
At each even longitudinal layer $2h$, consider the $n$ radius-$r_0$ disks around its retained waypoints. These disks are disjoint within a layer and across distinct even layers. The event that every disk in one layer contains a tree center has probability $p_n$, independently across layers. On that event all waypoints in the layer are blocked. Since every crossing visits every intervening even layer, any one of layers $2,4,\ldots,2N$ with this property prevents a crossing. This proves~\eqref{eq:blocklayerbound}. An independent sequence of events with probability $p_n>0$ contains such a layer almost surely. Depth monotonicity gives eventual blocking.
\end{IEEEproof}
Thus the depth-first, width-second current is zero even when $j_\infty^\nu>0$. Rapid joint depth growth can also produce zero current. For example, put $\beta_0=-\log(1-e^{-\zeta\pi r_0^2})$. If $N_n\ge e^{(\beta_0+\varepsilon)n}$ eventually for some $\varepsilon>0$, then~\eqref{eq:blocklayerbound} is summable in $n$, and $K_{N_n,n}=0$ eventually almost surely. The growth restriction in Theorem~\ref{thm:bulk} therefore cannot simply be omitted.

\section{Obstacle Intensity and Positive Bulk Current}\label{sec:intensity}
We now hold the corridor geometry, speed, and headway fixed and vary the obstacle intensity. Write $\kappa_\infty^\nu(\zeta)$ when this dependence needs to be explicit. A positive single-path survival probability is necessary for positive current by~\eqref{eq:comparison}. To prove positive current, we control the density of disjoint crossings directly.

\subsection{An Explicit Extinction Bound}
\begin{proposition}[high-intensity extinction]\label{prop:extinction}
For every $N\ge1$,
\begin{equation}
\kappa_N^\nu\le e^{-\zeta\pi R^2}
\min\left\{1,\left(5e^{-4\zeta Rd}\right)^N\right\}.\label{eq:extinctionbound}
\end{equation}
If $4\zeta Rd>\log5$, then $p_\infty^\nu=0$ and $j_\infty^\nu=0$. In this regime the normalized current vanishes for arbitrary jointly diverging depths and widths.
\end{proposition}
\begin{IEEEproof}
Between consecutive even layers, a route uses either one straight arc or two diagonal arcs. There are at most five choices: one straight choice and four ordered pairs of diagonal choices. Thus at most $5^N$ candidate paths connect a specified inlet to layer $2N$ in the full graph.

Every such path is a longitudinally monotone graph over $[0,L]$. At each longitudinal coordinate in this interval, its radius-$R$ neighborhood contains a vertical interval of length $2R$. The two endpoint disks contribute another two semicircles outside $[0,L]$. Hence the area of its forbidden region for tree centers is at least $2RL+\pi R^2$. The Poisson void probability and a union bound give
\begin{align}
\Prob(A_N(z_{0,0}))
&\le5^N e^{-\zeta(2RL+\pi R^2)}\notag\\
&=e^{-\zeta\pi R^2}\left(5e^{-4\zeta Rd}\right)^N.\label{eq:pathunion}
\end{align}
The event also requires a clear inlet, so its probability is at most $e^{-\zeta\pi R^2}$. The inlet-count argument of Proposition~\ref{prop:survival} gives $\kappa_N^\nu\le\Prob(A_N(z_{0,0}))$, proving~\eqref{eq:extinctionbound}. If $5e^{-4\zeta Rd}<1$, these probabilities tend to zero, and Proposition~\ref{prop:survival} proves the remaining assertions.
\end{IEEEproof}

\subsection{A Sparse-Forest Cut Bound}
Assume for the remainder of this subsection that diagonal corridors are available, so~\eqref{eq:availability} holds. Let
\begin{equation}
\ell_d=\sqrt{d^2+b^2},\quad A_d=2R\ell_d+\pi R^2,\quad
\epsilon_\zeta=1-e^{-\zeta A_d}.\label{eq:diagclosed}
\end{equation}
Thus $\epsilon_\zeta$ is the probability that a specified diagonal corridor is closed. Define
\begin{equation}
M=\left\lfloor\frac{2R+\ell_d}{\min\{d,b\}}\right\rfloor+1,
\quad \chi=2M^2.\label{eq:coloring}
\end{equation}
The constants are explicit but need not be small.

\begin{theorem}[sparse-forest current]\label{thm:sparse}
Fix $0<\delta\le1/2$ and put
\begin{equation}
\rho=6\epsilon_\zeta^{\delta/\chi}.\label{eq:rho}
\end{equation}
If $\rho<1$, then for all $N\ge1$ and $n\ge2$,
\begin{equation}
\Prob\left\{K_{N,n}<(1-\delta)\left(n-\tfrac32\right)\right\}
\le\frac{4N}{3(1-\rho)}\rho^{2n-3}.\label{eq:sparsebound}
\end{equation}
In particular,
\begin{equation}
\kappa_\infty^\nu(\zeta)\ge1-\delta.\label{eq:sparsecurrent}
\end{equation}
Moreover,
\begin{equation}
\lim_{\zeta\downarrow0}\kappa_\infty^\nu(\zeta)=1,
\quad
\lim_{\zeta\downarrow0}j_\infty^\nu(\zeta)=\tfrac{1}{2bH}.\label{eq:sparselimit}
\end{equation}
\end{theorem}
\begin{IEEEproof}
First discard all straight arcs. Give each open diagonal arc unit edge capacity, without imposing waypoint capacities, and let $F_{N,n}$ be the resulting maximum edge flow. An interior waypoint has at most two incoming and two outgoing diagonal arcs. Each inlet and outlet also has at most two incident diagonal arcs in the relevant direction. Dividing any feasible edge flow by two therefore respects unit waypoint capacities, including at the terminals. Max-flow integrality for the waypoint network gives
\begin{equation}
K_{N,n}\ge\tfrac12 F_{N,n}.\label{eq:edgecomparison}
\end{equation}

For the diagonal-only graph, join consecutive inlet vertices and consecutive outlet vertices by zero-capacity vertical chains. Use the dual crossing lengths in~\eqref{eq:duallengths}, with no exits through these terminal chains. The bounded faces are diamonds and terminal triangles. The diamond with vertices $z_{i-1,j},z_{i,j-1},z_{i+1,j},z_{i,j+1}$ has center
\begin{equation*}
(id,b(j+\tfrac12)),\quad i+j\text{ odd},
\end{equation*}
where $1\le i\le2N-1$ and $1\le j\le2n-2$. For $0\le k\le n-2$, the left terminal triangle has vertices $z_{0,2k},z_{1,2k+1},z_{0,2k+2}$; the right terminal triangle has vertices $z_{2N,2k},z_{2N-1,2k+1},z_{2N,2k+2}$. Choose their dual representatives at
\begin{equation*}
(\tau_0,b(2k+\tfrac32)),\quad
(L-\tau_0,b(2k+\tfrac32)),
\end{equation*}
respectively, for a fixed $0<\tau_0<d$. Use the diamond centers for all other representatives. Adjacent internal faces then have representative heights differing by exactly $b$. A downward crossing has the capacity of the corresponding diagonal as its cost, and an upward crossing has cost zero.

The $N$ upper attachment faces are the diamonds with centers $((2h+1)d,b(2n-\tfrac32))$, $0\le h<N$. The lower attachments are the diamonds with centers $(2hd,3b/2)$, $1\le h<N$, together with the two terminal triangles at $k=0$. Thus all lower attachment representatives have height $3b/2$, and an internal upper-to-lower path has net downward displacement
\begin{equation}
hb,\quad h=2n-3.\label{eq:dualheight}
\end{equation}
The port-to-port minimum passage cost equals $F_{N,n}$. Removing its first and last crossings leaves an internal path of no greater cost. Nonnegative lengths permit loops to be erased.

Consider a simple internal path of $\ell$ steps, with $D$ downward steps and $U$ upward steps. Then $D-U=h$ and $D+U=\ell$, so $D=(\ell+h)/2$ and $\ell\ge h$. If its cost is less than $(1-\delta)h$, the number $Z$ of closed downward crossings satisfies
\begin{align}
Z&>\tfrac{\ell+h}{2}-(1-\delta)h\notag\\
 &=\tfrac{\ell-h}{2}+\delta h\ge\delta\ell.\label{eq:closedfraction}
\end{align}
The last inequality uses $\delta\le1/2$. Thus a cheap cut must encounter a substantial fraction of closed diagonals.

We next account for dependence. Color a diagonal by its orientation and by the two coordinates of its tail modulo $M$. This gives at most $\chi$ colors. Distinct diagonals of one color are parallel translates by a vector of length at least $M\min\{d,b\}>2R+\ell_d$. Their centerline segments are therefore more than $2R$ apart, so their clearance capsules are disjoint. Within each color, corridor states are independent. Every collection of $k$ distinct diagonals contains at least $k/\chi$ diagonals of one color, and consequently
\begin{equation}
\Prob\{\text{all diagonals in the collection are closed}\}
\le\epsilon_\zeta^{k/\chi}.\label{eq:dependentbound}
\end{equation}
For a prescribed simple dual path, the crossed primal edges are distinct. A union bound over subsets of its downward crossings, using~\eqref{eq:dependentbound}, gives
\begin{equation}
\Prob\{Z\ge\delta\ell\}\le2^\ell
\epsilon_\zeta^{\delta\ell/\chi}.\label{eq:pathclosed}
\end{equation}

The internal dual has degree at most four. The number of simple $\ell$-step paths starting at one of the $N$ upper attachments is at most $4N3^{\ell-1}$ for $\ell\ge1$. On the event in~\eqref{eq:sparsebound},~\eqref{eq:edgecomparison} implies $F_{N,n}<(1-\delta)h$, so there exists a path satisfying~\eqref{eq:closedfraction}. Summing~\eqref{eq:pathclosed} over all possible lengths therefore yields
\begin{equation*}
\Prob\{F_{N,n}<(1-\delta)h\}
\le\tfrac{4N}{3}\sum_{\ell\ge h}\left(6\epsilon_\zeta^{\delta/\chi}\right)^\ell,
\end{equation*}
which proves~\eqref{eq:sparsebound}.

Take $N=n$. The right side of~\eqref{eq:sparsebound} is summable, so Borel--Cantelli gives
\begin{equation}
\liminf_{n\to\infty}\frac{K_{n,n}}n\ge1-\delta
\quad\text{almost surely}.\label{eq:sparsediagonal}
\end{equation}
For every fixed depth $M_0$, depth monotonicity gives $K_{M_0,n}\ge K_{n,n}$ once $n\ge M_0$. Theorem~\ref{thm:transverse} implies $\kappa_{M_0}^\nu\ge1-\delta$. Let $M_0\to\infty$ to obtain~\eqref{eq:sparsecurrent}. Finally, $\epsilon_\zeta\to0$ as $\zeta\downarrow0$. For each fixed $\delta>0$, condition~\eqref{eq:rho} therefore holds at all sufficiently small intensities. Since $\kappa_\infty^\nu\le1$, letting $\delta\downarrow0$ proves~\eqref{eq:sparselimit}.
\end{IEEEproof}
The proof establishes positive current without invoking the shape theorem. The latter identifies the same current with the vertical time constant and with the joint-growth limits. Dependence affects the coloring exponent $\chi$; it is not replaced by an independent-edge approximation.

\subsection{A Nontrivial Critical Intensity}
\begin{corollary}[intensity threshold for extensive current]\label{cor:critical}
When diagonal corridors are available, define
\begin{equation}
\zeta_c=\sup\{\zeta>0:\kappa_\infty^\nu(\zeta)>0\}.\label{eq:criticaldef}
\end{equation}
Then, with the constants in~\eqref{eq:diagclosed}--\eqref{eq:coloring},
\begin{equation}
0<\zeta_0:=-\frac{\log(1-6^{-2\chi})}{A_d}
\le\zeta_c\le\frac{\log5}{4Rd}.\label{eq:criticalbounds}
\end{equation}
The bulk current is positive for $0<\zeta<\zeta_c$ and zero for $\zeta>\zeta_c$. No assertion at $\zeta=\zeta_c$ is needed for these conclusions.
\end{corollary}
\begin{IEEEproof}
Poisson superposition couples two intensities by adding tree centers. Adding obstacles cannot open a corridor or increase $K_{N,n}$, so $\kappa_\infty^\nu(\zeta)$ is nonincreasing. Theorem~\ref{thm:sparse} with $\delta=1/2$ applies whenever $\zeta<\zeta_0$, so the set in~\eqref{eq:criticaldef} contains $(0,\zeta_0)$. Proposition~\ref{prop:extinction} makes it bounded above by $\log5/(4Rd)$. Monotonicity then gives the two regimes.
\end{IEEEproof}
The estimates give sufficient conditions, not an exact value of the threshold. In the straight-only regime, $j_\infty^\nu(\zeta)=0$ for every positive intensity by~\eqref{eq:straightonly}.

\section{Conclusion}\label{sec:conclusion}
Karaman and Frazzoli's stochastic-forest theorem gives a fundamental speed limit for indefinite flight by one bird. For many birds crossing a forest belt, the corresponding engineering question is a current limit, which we establish in several regimes. With a prescribed forward corridor library and common waypoint headway, that current equals a vertex-cut capacity. Repeated launches along an integral maximum flow's disjoint paths achieve the cut bound while preserving constant forward speed. Because all routes have the same transit time, the same schedule maximizes deliveries at every finite horizon. The cut converse and trajectory schedule therefore concern the same model.

The analogy we make with porous media concerns geometry-limited transport. Unlike Darcy flow~\cite{bear1972,darcy1856}, the present current is not related to a pressure gradient through an effective permeability. Its principal object is a capacity-limited throughput, governed by a dynamically feasible cut rather than viscous dissipation. The graph exactly models the specified protocol, not all continuous multi-agent trajectories.

The Poisson forest admits a deterministic transverse current-density law, a finite-width mean-error bound, and concentration about the mean. Current is strictly positive at every finite depth. A uniform asymmetric dual shape theorem and a confined vertical-block construction identify the joint-growth current with the width-first, depth-second limit whenever the diverging depth is of order at most the width. Lateral symmetry removes dependence on the positive aspect ratio. At fixed finite width, however, a sufficiently deep forest is eventually blocked almost surely.

When diagonal corridors are available, sufficiently sparse forests support positive extensive current. As obstacle intensity vanishes, this current approaches the inlet-capacity limit. A forward-path count gives extinction at sufficiently large intensity, establishing a nontrivial critical intensity for the protocol. Determining the critical value, the behavior at criticality, and the relationship between positive survival and positive current are further probabilistic questions.

Congestion is governed not only by average tree density but by bottlenecks in the feasible-corridor network. Below the cut limit, repeated-path flight and admission sustain traffic; above it, no policy respecting the protocol can avoid overload. Quantifying information, tracking, and coordination costs would determine how closely less structured controllers approach this benchmark.

Alternative path objectives become relevant when flight time, rather than only throughput, varies across routes. The pipelined transport model of Marla et al.~\cite{marla2022} requires commodities to traverse single paths and identifies a short-and-wide path criterion. This provides a related perspective on physical transport without splitting individual agents. In the present model every crossing takes exactly $L/\nu$, so there is no route-dependent transit-time tradeoff. Such a tradeoff would arise with nonuniform forward speeds, route-dependent dwell times, or other path costs.

Allowing safe storage inside the forest gives a different operational extension. With interacting corridor services, the relevant capacity region is determined by jointly feasible service vectors. Backpressure can then provide a throughput-oriented policy under its queueing and service assumptions~\cite{neely2010,tassiulas1992,srikant2014}. Such a model permits intermediate waiting and must account for actual transfers, storage resources, and merge conflicts. Density-dependent maneuverability or service rates would add a further nonlinear constraint to this network control problem.

\section*{AI Use}
ChatGPT 5.6-Sol and then ChatGPT 6 were used to support mathematical development, writing, and figure drawing.

\section*{Acknowledgment}
Initial discussions with Sertac Karaman are appreciated.

\bibliographystyle{IEEEtran}
\bibliography{flock_flight}

\begin{thebibliography}{10}
\providecommand{\url}[1]{#1}
\csname url@samestyle\endcsname
\providecommand{\newblock}{\relax}
\providecommand{\bibinfo}[2]{#2}
\providecommand{\BIBentrySTDinterwordspacing}{\spaceskip=0pt\relax}
\providecommand{\BIBentryALTinterwordstretchfactor}{4}
\providecommand{\BIBentryALTinterwordspacing}{\spaceskip=\fontdimen2\font plus
\BIBentryALTinterwordstretchfactor\fontdimen3\font minus
  \fontdimen4\font\relax}
\providecommand{\BIBforeignlanguage}[2]{{%
\expandafter\ifx\csname l@#1\endcsname\relax
\typeout{** WARNING: IEEEtran.bst: No hyphenation pattern has been}%
\typeout{** loaded for the language `#1'. Using the pattern for}%
\typeout{** the default language instead.}%
\else
\language=\csname l@#1\endcsname
\fi
#2}}
\providecommand{\BIBdecl}{\relax}
\BIBdecl

\bibitem{karaman2012}
S.~Karaman and E.~Frazzoli, ``High-speed flight in an ergodic forest,'' in
  \emph{Proc. IEEE Int. Conf. Robotics and Automation}, May 2012, pp.
  2899--2906.

\bibitem{bear1972}
J.~Bear, \emph{Dynamics of Fluids in Porous Media}.\hskip 1em plus 0.5em minus
  0.4em\relax New York, NY, USA: American Elsevier, 1972.

\bibitem{chayes1986}
J.~T. Chayes and L.~Chayes, ``Bulk transport properties and exponent
  inequalities for random resistor and flow networks,'' \emph{Communications in
  Mathematical Physics}, vol. 105, pp. 133--152, 1986.

\bibitem{ford1956}
L.~R. Ford, Jr. and D.~R. Fulkerson, ``Maximal flow through a network,''
  \emph{Canadian Journal of Mathematics}, vol.~8, pp. 399--404, 1956.

\bibitem{elias1956}
P.~Elias, A.~Feinstein, and C.~E. Shannon, ``A note on the maximum flow through
  a network,'' \emph{IRE Transactions on Information Theory}, vol.~2, no.~4,
  pp. 117--119, Dec. 1956.

\bibitem{kesten1987}
H.~Kesten, ``Surfaces with minimal random weights and maximal flows: A higher
  dimensional version of first-passage percolation,'' \emph{Illinois Journal of
  Mathematics}, vol.~31, no.~1, pp. 99--166, 1987.

\bibitem{rossignol2010tilted}
R.~Rossignol and M.~Th{\'e}ret, ``Law of large numbers for the maximal flow
  through tilted cylinders in two-dimensional first passage percolation,''
  \emph{Stochastic Processes and their Applications}, vol. 120, no.~6, pp.
  873--900, 2010.

\bibitem{rossignol2010box}
------, ``Lower large deviations and laws of large numbers for maximal flows
  through a box in first passage percolation,'' \emph{Annales de l'Institut
  Henri Poincar{\'e}, Probabilit{\'e}s et Statistiques}, vol.~46, no.~4, pp.
  1093--1131, 2010.

\bibitem{farias2020}
R.~Farias and M.~Kallmann, ``Planar max flow maps and determination of lanes
  with clearance,'' \emph{Autonomous Robots}, vol.~44, pp. 1323--1339, 2020.

\bibitem{ford1958}
L.~R. Ford, Jr. and D.~R. Fulkerson, ``Constructing maximal dynamic flows from
  static flows,'' \emph{Operations Research}, vol.~6, no.~3, pp. 419--433,
  1958.

\bibitem{kingman1973}
J.~F.~C. Kingman, ``Subadditive ergodic theory,'' \emph{The Annals of
  Probability}, vol.~1, no.~6, pp. 883--909, 1973.

\bibitem{steele1989}
J.~M. Steele, ``{Kingman's} subadditive ergodic theorem,'' \emph{Annales de
  l'Institut Henri Poincar{\'e}, Probabilit{\'e}s et Statistiques}, vol.~25,
  no.~1, pp. 93--98, 1989.

\bibitem{neely2010}
M.~J. Neely, \emph{Stochastic Network Optimization with Application to
  Communication and Queueing Systems}.\hskip 1em plus 0.5em minus 0.4em\relax
  San Rafael, CA, USA: Morgan \& Claypool, 2010.

\bibitem{last2017}
G.~Last and M.~Penrose, \emph{Lectures on the Poisson Process}.\hskip 1em plus
  0.5em minus 0.4em\relax Cambridge, UK: Cambridge University Press, 2017.

\bibitem{ahuja1993}
R.~K. Ahuja, T.~L. Magnanti, and J.~B. Orlin, \emph{Network Flows: Theory,
  Algorithms, and Applications}.\hskip 1em plus 0.5em minus 0.4em\relax
  Englewood Cliffs, NJ, USA: Prentice Hall, 1993.

\bibitem{mcdiarmid1989}
C.~McDiarmid, ``On the method of bounded differences,'' in \emph{Surveys in
  Combinatorics, 1989}, ser. London Mathematical Society Lecture Note Series,
  J.~Siemons, Ed.\hskip 1em plus 0.5em minus 0.4em\relax Cambridge, UK:
  Cambridge University Press, 1989, vol. 141, pp. 148--188.

\bibitem{hoeffding1963}
W.~Hoeffding, ``Probability inequalities for sums of bounded random
  variables,'' \emph{Journal of the American Statistical Association}, vol.~58,
  no. 301, pp. 13--30, 1963.

\bibitem{lighthill1955}
M.~J. Lighthill and G.~B. Whitham, ``On kinematic waves. {II}. a theory of
  traffic flow on long crowded roads,'' \emph{Proceedings of the Royal Society
  of London. Series A}, vol. 229, no. 1178, pp. 317--345, 1955.

\bibitem{richards1956}
P.~I. Richards, ``Shock waves on the highway,'' \emph{Operations Research},
  vol.~4, no.~1, pp. 42--51, 1956.

\bibitem{kaplan2011}
H.~Kaplan and Y.~Nussbaum, ``Maximum flow in directed planar graphs with vertex
  capacities,'' \emph{Algorithmica}, vol.~61, no.~1, pp. 174--189, 2011.

\bibitem{hassin1981}
R.~Hassin, ``Maximum flow in $(s,t)$ planar networks,'' \emph{Information
  Processing Letters}, vol.~13, no.~3, p. 107, 1981.

\bibitem{erickson2010}
J.~Erickson, ``Maximum flows and parametric shortest paths in planar graphs,''
  in \emph{Proc. 21st Annual ACM--SIAM Symp. Discrete Algorithms}, 2010, pp.
  794--804.

\bibitem{zenklusen2010}
R.~Zenklusen, ``Network flow interdiction on planar graphs,'' \emph{Discrete
  Applied Mathematics}, vol. 158, no.~13, pp. 1441--1455, 2010.

\bibitem{boivin1990}
D.~Boivin, ``First passage percolation: The stationary case,''
  \emph{Probability Theory and Related Fields}, vol.~86, no.~4, pp. 491--499,
  1990.

\bibitem{haggstrom1995}
O.~H{\"a}ggstr{\"o}m and R.~Meester, ``Asymptotic shapes for stationary first
  passage percolation,'' \emph{The Annals of Probability}, vol.~23, no.~4, pp.
  1511--1522, 1995.

\bibitem{bjorklund2010}
M.~Bj{\"o}rklund, ``The asymptotic shape theorem for generalized first passage
  percolation,'' \emph{The Annals of Probability}, vol.~38, no.~2, pp.
  632--660, 2010.

\bibitem{krengel1985}
U.~Krengel, \emph{Ergodic Theorems}, ser. De Gruyter Studies in
  Mathematics.\hskip 1em plus 0.5em minus 0.4em\relax Berlin, Germany: Walter
  de Gruyter, 1985, vol.~6.

\bibitem{darcy1856}
H.~Darcy, \emph{Les fontaines publiques de la ville de Dijon}.\hskip 1em plus
  0.5em minus 0.4em\relax Paris, France: Dalmont, 1856.

\bibitem{marla2022}
L.~Marla, L.~R. Varshney, D.~Shah, N.~A. Prakash, and M.~E. Gale, ``Short and
  wide network paths,'' \emph{IEEE Transactions on Network Science and
  Engineering}, vol.~9, no.~2, pp. 524--537, Mar. 2022.

\bibitem{tassiulas1992}
L.~Tassiulas and A.~Ephremides, ``Stability properties of constrained queueing
  systems and scheduling policies for maximum throughput in multihop radio
  networks,'' \emph{IEEE Transactions on Automatic Control}, vol.~37, no.~12,
  pp. 1936--1948, 1992.

\bibitem{srikant2014}
R.~Srikant and L.~Ying, \emph{Communication Networks: An Optimization, Control,
  and Stochastic Networks Perspective}.\hskip 1em plus 0.5em minus 0.4em\relax
  Cambridge, UK: Cambridge University Press, 2014.

\end{thebibliography}
\end{document}